\documentclass[twoside,10pt]{article}
\usepackage{amsthm,amsmath,amssymb,amscd,enumerate}
\usepackage{amsfonts}
\usepackage{graphicx}
\usepackage{fancyhdr}
\begin{document}
\begin{center} {\Large\bf Optimal Hamiltonian of Fermion Flows}
\vskip 1cm {\bf   Luigi Accardi$^{1}$, Andreas Boukas$^{2}$}\\
\
\\$^{1}$ Centro Vito Volterra\\
Universit\`{a} di Roma Tor Vergata \\
via Columbia, 2-- 00133 Roma, Italia\\
e-mail: accardi@Volterra.mat.uniroma2.it\\
\
\\$^{2}$ Department of Mathematics and Natural Sciences\\
American College of Greece\\
Aghia Paraskevi, Athens 15342, Greece\\
e-mail: andreasboukas@acgmail.gr\\
\end{center}

\begin{abstract}
After providing a general formulation of Fermion flows within the context of Hudson-Parthasarathy quantum stochastic calculus, we consider the problem of determining the noise coefficients of the Hamiltonian associated with a Fermion flow so as to minimize a naturally associated quadratic performance functional. This extends to Fermion flows results  of the authors previously obtained for Boson flows .

\end{abstract}
\numberwithin{equation}{section}
\newtheorem{theorem}{Theorem}
\newtheorem{lemma}{Lemma}
\newtheorem{proposition}{Proposition}
\newtheorem{corollary}{Corollary}
\newtheorem{example}{Example}
\newtheorem{algorithm}{Algorithm}
\newtheorem*{main}{Main~Theorem}
\newtheorem{notation}{Notation}
\newtheorem{remark}{Remark}
\newtheorem{definition}{Definition}

fermion.tex

\section{\textbf{Introduction}}

In the Heisenberg picture of quantum mechanics, the time-evolution of an observable $X$ is described  by the operator process  $j_t(X)=U^*(t)\,X\,U(t)$ where $U(t)=e^{-itH}$ and the Hamiltonian $H$ is a self-adjoint operator on the wave function Hilbert space. The unitary process $U(t)$ and the self-adjoint process $j_t(X)$ satisfy, respectively,

\begin{eqnarray*}
dU(t)=-iH\,U(t)\,dt,\,\,U(0)=I
\end{eqnarray*}

and 

\begin{eqnarray*}
dj_t(X)=i[H,j_t(X)]\,dt,\,\,j_0(X)=X
\end{eqnarray*}

where 

\begin{eqnarray*}
[H,j_t(X)]:=H\,j_t(X)-j_t(X)\,H
\end{eqnarray*}

In the presence of quantum noise, the equation satisfied by  $U(t)$ is replaced by a Hudson-Parthasarathy quantum stochastic differential equation  driven by that noise (see \cite{12} and \cite{14}) and the corresponding equation for the quantum flow  $j_t(X)$  is interpreted as the  Heisenberg picture of the Schr\"{o}dinger equation in the presence of noise. The emergence of such equations as stochastic limits of classical Schr\"{o}dinger equations is described in \cite{6a}. From the point of view of quantum control theory, the problem of minimizing a quadratic performance functional associated with a Hudson-Parthasarathy quantum flow driven by Boson quantum noise, has been considered in \cite{1}, 
\cite{4}-\cite{6} and \cite{7}-\cite{11}.  In this paper we consider the same problem for quantum flows driven by Fermion flows. Our approach is based on the representation of the Fermion commutation relations on  Boson Fock space obtained in \cite{0} and \cite{13}. A unified approach  to the quadratic cost control of quantum processes driven by Boson, Fermion, Finite-Difference and a wide class of other quantum noises,  can be found in \cite{2} with the use of the representation free quantum stochastic calculus of \cite{6b}.

This paper is structured as follows: In section 2 we describe the main features of Hudson-Parthasarathy calculus and the representation of the Fermion commutation relations on  Boson Fock space. In section 3 we obtain the quantum stochastic  differential equations satisfied by Fermion flows. In section 4 we obtain a new algebraic formulation of Fermion evolution equations and  flows and we describe the equations satisfied by their structure maps. Finally, in section 5 we define the quadratic performance functionals associated with Fermion flows and we obtain the coefficients of the corresponding Hamiltonian for which these functionals are minimized. 

\section{Quantum Stochastic Calculus }
The Boson Fock space $\Gamma:=\Gamma(L^2(\mathbb{R_+},\mathbb{C}))$
over $L^2(\mathbb{R_+},\mathbb{C}) $  is the Hilbert space
completion of the linear span of the exponential vectors  $\psi(f)$
under the inner product
\begin{eqnarray*}
<\psi(f),\psi(g)>:=e^{<f,g>}
\end{eqnarray*}where
$f,g \in L^2(\mathbb{R_+},\mathbb{C})$ and
$<f,g>=\int_0^{+\infty}\,\bar{f}(s)\,g(s)\,ds$ where, here and in
what follows,  $\bar{z}$ denotes the complex conjugate of $z\in
\mathbb{C}$.  The
 annihilation, creation and conservation operator processes $A_t$, $A^{\dagger}_t$ and $\Lambda_t$ respectively, are defined on the exponential vectors $\psi(g)$ of $\Gamma$ by\begin{eqnarray*}A_t\psi(g)&:=&\int_0^t\,g(s)\,ds\,\,\psi(g)\\A^{\dagger}_t\psi(g)&:=&\frac{\partial}{\partial \epsilon}|_{\epsilon=0}\,\psi(g+\epsilon \chi_{[0,t]})\\\Lambda_t\psi(g)&:=&\frac{\partial}{\partial \epsilon}|_{\epsilon=0}\,\psi(e^{\epsilon \chi_{[0,t]})}g).\end{eqnarray*}The basic   quantum stochastic differentials $dA_t$, $dA^{\dagger}_t$, and $d\Lambda_t$ are defined by\begin{eqnarray*}&dA_t:=A_{t+dt}-A_t& \\&dA^{\dagger}_t: = A^{\dagger}_{t+dt}-A^{\dagger}_{t}  &\\&d \Lambda_t:= \Lambda_{t+dt} -\Lambda_t .&\end{eqnarray*}Hudson and Parthasarathy defined in \cite{12} stochastic integration with respect to the noise differentials $dA_t$, $dA^{\dagger}_t$ and $d \Lambda_t$ and obtained the  It\^{o} multiplication
 table\begin{center}
 \begin{tabular}
 {c|cccc}$\cdot$&$dA_t^{\dagger}$&$d\Lambda_t$&$dA_t$&$dt$\\\hline$dA_t^{\dagger}$&$0$&$0$&$0$&$0$\\$d\Lambda_t$&$dA_t^{\dagger}$&$d\Lambda_t$&$0$&$0$\\$dA_t$&$dt$&$dA_t$&$0$&$0$\\$dt$&$0$&$0$&$0$&$0$.\end{tabular}\end{center}We couple $\Gamma$ with a "system" Hilbert space $\mathcal{H}$ and consider processes defined on $\mathcal{H} \otimes \Gamma$. The fundamental theorems of the Hudson-Parthasarathy quantum stochastic calculus  give formulas for expressing the matrix elements of quantum stochastic integrals  in terms of ordinary Riemann-Lebesgue integrals. \begin{theorem}Let \begin{eqnarray*}&M(t):=\int_0^t\,E(s)\,d\Lambda_s+F(s)\,dA_s+G(s)\,dA^{\dagger}_s+H(s)\,ds&\end{eqnarray*}where $E$, $F$, $G$, $H$ are (in general) time dependent adapted processes. Let also $u\otimes \psi (f)$ and $u \otimes \psi (g)$ be in the "exponential domain" of $\mathcal{H} \otimes \Gamma$. Then\begin{eqnarray*}&<u\otimes \psi (f),M(t)\,  u \otimes \psi
 (g)>=&\\&\int_0^t <u \otimes \psi (f),\left(\bar{f}(s)\,g(s) \,E(s)+g(s)\,F(s)+\bar{f}(s)\,G(s)+H(s)\right) u \otimes \psi (g)>\,ds.&\nonumber\end{eqnarray*}\end{theorem} \begin{proof}See Theorem 4.1 of \cite{12}.\end{proof}\begin{theorem}Let \begin{eqnarray*}&M(t):=\int_0^t\,E(s)\,d\Lambda_s+F(s)\,dA_s+G(s)\,dA^{\dagger}_s+H(s)\,ds&\end{eqnarray*}and\begin{eqnarray*}&M^{\prime}(t):=\int_0^t\,E^{\prime}(s)\,d\Lambda_s+F^{\prime}(s)\,dA_s+G^{\prime}(s)\,dA^{\dagger}_s+H^{\prime}(s)\,ds&\end{eqnarray*}where $E$, $F$, $G$, $H$, $E^{\prime}$, $F^{\prime}$, $G^{\prime}$, $H^{\prime}$ are (in general) time dependent adapted processes. Let also $u\otimes \psi (f)$ and $u \otimes \psi (g)$ be in the exponential domain of $\mathcal{H} \otimes \Gamma$. Then\begin{eqnarray*}&<M(t) \,u \otimes \psi (f),M^{\prime} (t)\, u \otimes \psi (g)>=&\\&\int_0^t \{ <M(s)\,u \otimes \psi (f),\left( \bar{f}(s)\,g(s) E^{\prime} (s)+g(s)\,F^{\prime} (s)+\bar{f}(s)\,G^{\prime} (s)+H^{\prime} (s) \right) u
 \otimes \psi (g)>&\nonumber\\&  + < \left( \bar{g}(s)\,f(s)\, E(s) + f(s)\,F(s) +\bar{g}(s)\, G(s)+H(s) \right) u \otimes \psi (f),  M^{\prime} (s)\,u \otimes \psi (g)>&\nonumber\\&+   < \left( f(s) E(s)+G(s)   \right) u \otimes \psi (f),  \left( g(s) E^{\prime} (s)+G^{\prime} (s)   \right)  u \otimes \psi (g)>\}\, ds .&\nonumber\end{eqnarray*}\end{theorem} \begin{proof}See Theorem 4.3 of \cite{12}.\end{proof}The connection between classical and quantum stochastic analysis is given in the following:\begin{theorem} The processes $B=\{B_t\,,\,t \geq 0\}$ and $P=\{P_t\,,\,t \geq 0\}$ defined by\begin{eqnarray*}&B_t:=A_t+A_t^{\dagger}&\end{eqnarray*}and\begin{eqnarray*}&P_t:=\Lambda_t+\sqrt{\lambda}\,(A_t+A_t^{\dagger})+\lambda \,t&\end{eqnarray*}are identified with Brownian motion and  Poisson process of intensity $\lambda$ respectively, in the sense that  their vacuum characteristic functionals are given by\begin{eqnarray*}&<\psi (0),e^{i\,s\,B_t}\,\psi (0)>= e^{ - \frac{s^2}{ 2}\,t }
 &\end{eqnarray*}
 and
 \begin{eqnarray*}
 &<\psi (0),e^{i\,s\,P_t}\,\psi (0)>=e^{ \lambda\,\left(e^{i\,s}-1\right)\,t }.&
 \end{eqnarray*}
 \end{theorem}
 \begin{proof} See Theorem 5 of \cite{11}.
 \end{proof}
 The processes $A_t$, $A^{\dagger}_t$  satisfy the Boson  Commutation Relations
 \begin{eqnarray*}&\left[A_t,A^{\dagger}_t\right]:=A_t\,A^{\dagger}_t-A^{\dagger}_t\,A_t=t\,I.&
 \end{eqnarray*}In \cite{13} Hudson and Parthasarathy showed that the processes $F_t$ and $F^{\dagger}_t$ defined on the Boson Fock space by\begin{eqnarray*}F_t&:=&\int_0^t\,J_s\,dA_s   \\F^{\dagger}_t&:=&\int_0^t\,J_s\,dA_s^{\dagger}\end{eqnarray*}satisfy the Fermion anti-commutation relations\begin{eqnarray*}&\{F_t,F^{\dagger}_t\}:=F_t\,F^{\dagger}_t+F^{\dagger}_t\,F_t=t\,I.&\end{eqnarray*}It follows that\begin{eqnarray}dF_t&=&J_t\,dA_t\label{1}\\dF^{\dagger}_t&=&J_t\,dA_t^{\dagger}\label{2} .\end{eqnarray}Here $J_t$ is the self-adjoint, unitary-valued, adapted, so called "reflection" process, acting on the noise part
 of the Fock space and extended as the identity on the system part, defined by
 \begin{eqnarray*}
 J_t:=\gamma(-P_{[0,t]}+P_{(t,+\infty})
 \end{eqnarray*}
 where $P_S$ denotes the multiplication operator by $\chi_S$ and $\gamma$ is the second quantization operator defined by
 \begin{eqnarray*}
 \gamma(U)\,\psi(f):=\psi(U\,f).
 \end{eqnarray*}
 The reflection process  $J_t$ commutes with system space operators and satisfies the differential
 equation (cf. Lemma 3.1 of \cite{13})
 \begin{eqnarray}dJ_t&=&-2\,J_t\,d\Lambda_t\label{3}\\J_0&=&1.\nonumber
 \end{eqnarray}
 \section{Fermion Evolutions and Flows}
 As shown in \cite{13}, see also \cite{0},  Fermion unitary evolution equations have the form
 \begin{eqnarray}
 dU_t=-\left(\left(iH+\frac{1}{2}\,L^*L\right)\,dt+ L^* \,W\,dF_t -L\, dF_t^{\dagger}+\left(1-W\right)\,d\Lambda_t\right)U_t\label{4}\end{eqnarray}with adjoint\begin{eqnarray}dU_t^*=-U_t^*\,((-iH+\frac{1}{2}\,L^*L)\,dt- L^* dF_t +W^*\,L\, dF_t^{\dagger}+(1-W^*)\,d\Lambda_t)\label{5}
  \end{eqnarray}
 where, $U_0=U^*_0=1$ and, for each $t\geq0$, $U_t$ is a unitary operator defined on the tensor product
 $\mathcal{H} \otimes \Gamma(L^2(\mathbb{R}_+,\mathcal{ \mathbb{C}  }))$ of the system Hilbert space $\mathcal{H} $ and
 the noise Fock space $\Gamma$.  Here $H$, $L$, $W$ are in $\mathcal{B}(\mathcal{H})$, the space of bounded linear
 operators on $\mathcal{H} $, with $W$ unitary and $H$ self-adjoint.   We identify time-independent, bounded, system
 space operators $X$ with their ampliation $X \otimes 1$ to
 $\mathcal{H} \otimes \Gamma(L^2(\mathbb{R}_+,\mathcal{\mathbb{C}}))$.
 \begin{proposition}\label{prop1}
 Let \begin{eqnarray*}
 \phi_t(T,S):=V_t^*\,(T+S\,J_t)\,V_t
 \end{eqnarray*}
 where $T,S$ are bounded system space operators and $V_t,V_t^*$ satisfy the quantum
 stochastic differential equations
 \begin{eqnarray*}
 dV_t&=&\left(\alpha\,dt+ \beta\,dF_t +\gamma\, dF_t^{\dagger}+\delta\,d\Lambda_t\right)V_t \\
&=&\left(\alpha\,dt+ \beta\,J_t\,dA_t +\gamma\,J_t\,
 dA_t^{\dagger}+\delta\,d\Lambda_t\right)V_t
 \end{eqnarray*}
 and
 \begin{eqnarray*}
 dV_t^*&=&V_t^*\left(\alpha^*\,dt+ \beta^*\,dF_t^{\dagger} +\gamma^*\, dF_t
 + \delta^*\,d\Lambda_t\right)\\
&=&V_t^*\left(\alpha^*\,dt+ \beta^*\,J_t\,dA_t^{\dagger}
 +\gamma^*\,J_t\, dA_t+\delta^*\,d\Lambda_t\right)
 \end{eqnarray*}
 where $\alpha,\beta,\gamma,\delta$ are bounded system space operators. Then
 \begin{eqnarray}
 d\phi_t(T,S)&=&\phi_t(\alpha^*\,T+T\,\alpha+\gamma^*\,T\,\gamma,\alpha^*\,S+S\,\alpha+\gamma^*\,S\,\gamma)\,dt\label{6}\\&+&\phi_t(\gamma^*\,S+S\,\beta-\gamma^*\,S\,(2+\delta),\gamma^*\,T+T\,\beta+\gamma^*\,T\,\delta)\,dA_t\nonumber\\  &+&\phi_t(\beta^*\,S+S\,\gamma-\delta^*\,S\,\gamma,\beta^*\,T+T\,\gamma+\delta^*\,T\,\gamma)\,dA_t^{\dagger}\nonumber\\&+&\phi_t(\delta^*\,T+T\,\delta+\delta^*\,T\,\delta,-(2\,S+S\,\delta+\delta^*\,S+\delta^*\,S\,\delta))\,d\Lambda_t\nonumber\end{eqnarray}with\begin{eqnarray}\phi_0(T,S)=T+S\label{6a}\end{eqnarray}\end{proposition}\begin{proof} Making use of the
 algebraic rule
 \begin{eqnarray*}
 d(x\,y)=dx\, y+x\, dy+dx \, dy
 \end{eqnarray*}we find
 \begin{eqnarray*}
 &d\phi_t(T,S)=dV^*_t\,(T+S\,J_t)\,V_t +V^*_t\,d((T+S\,J_t)\,V_t)+dV^*_t\,d((T+S\,J_t)\,V_t)&\\&=dV^*_t\,(T+S\,J_t)\,V_t+V^*_t\,T\,dV_t+V^*_t\,S\,d(J_t\,V_t)+dV^*_t\,T\,dV_t+dV^*_t\,S\,d(J_t\,V_t)&.\end{eqnarray*}But, by (\ref{3}), the It\^{o} table for the Boson stochastic differentials and the fact that $ J_t^2=1$\begin{eqnarray*}d(J_t\,V_t)&=&dJ_t\, V_t+J_t\, dV_t+dJ_t \, dV_t\\&=&-2\,J_t\,d\Lambda_t\,V_t+J_t\,\left(\alpha\,dt+ \beta\,J_t\,dA_t +\gamma\,J_t \,dA_t^{\dagger}+\delta\,d\Lambda_t\right)\,V_t\\&-&2\,J_t\,d\Lambda_t\,\left(\alpha\,dt+ \beta\,J_t\,dA_t +\gamma\,J_t \,dA_t^{\dagger}+\delta\,d\Lambda_t\right)\,V_t\\&=&\left(\alpha\,J_t\,dt+\beta\,dA_t -\gamma\,dA_t^{\dagger}-(2+\delta)\,J_t\,d\Lambda_t\right)\,V_t.\end{eqnarray*}Thus

\begin{eqnarray*}
&d\phi_t(T,S)=V^*_t\,\left(\alpha^*\,dt+ \beta^*\,J_t\,dA_t^{\dagger} +\gamma^*\,J_t\,
 dA_t+\delta^*\,d\Lambda_t\right)\,(T+S\,J_t)\,V_t&\\
&+V^*_t\,T\,\left(\alpha\,dt+ \beta\,J_t\,dA_t +\gamma\,J_t \,dA_t^{\dagger}+\delta\,d\Lambda_t\right)\,V_t&\\
&+V^*_t\,S\,\left(\alpha\,J_t\,dt+\beta\,dA_t -\gamma\,dA_t^{\dagger}-(2+\delta)\,J_t\,d\Lambda_t\right)\,V_t&\\
&+V^*_t\,\left(\alpha^*\,dt+ \beta^*\,J_t\,dA_t^{\dagger} +\gamma^*\,J_t\, dA_t+\delta^*\,d\Lambda_t\right)&\\
&\times\,T\,\left(\alpha\,dt+ \beta\,J_t\,dA_t +\gamma\,J_t \,dA_t^{\dagger}+\delta\,d\Lambda_t\right)\,V_t&\\
&+V^*_t\,\left(\alpha^*\,dt+ \beta^*\,J_t\,dA_t^{\dagger} +\gamma^*\,J_t\, dA_t+\delta^*\,d\Lambda_t\right)&\\
&\times \,S\,\left(\alpha\,J_t\,dt+\beta\,dA_t
 -\gamma\,dA_t^{\dagger}-(2+\delta)\,J_t\,d\Lambda_t\right)\,V_t&\\
&=\phi_t(\alpha^*\,T+T\,\alpha+\gamma^*\,T\,\gamma,\alpha^*\,S+S\,\alpha+\gamma^*\,S\,\gamma)\,dt&\\
&+\phi_t(\gamma^*\,S+S\,\beta-\gamma^*\,S\,(2+\delta),\gamma^*\,T+T\,\beta+\gamma^*\,T\,\delta)\,dA_t&\\
&+\phi_t(\beta^*\,S+S\,\gamma-\delta^*\,S\,\gamma,\beta^*\,T+T\,\gamma+\delta^*\,T\,\gamma)\,dA_t^{\dagger}&\\
&+\phi_t(\delta^*\,T+T\,\delta+\delta^*\,T\,\delta,-(2\,S+S\,\delta+\delta^*\,S+\delta^*\,S\,\delta))\,d\Lambda_t&
\end{eqnarray*}

\end{proof}With the processes $U_t$  and $U_t^*$ of  (\ref{4}) and (\ref{5})  we associate the Fermion flow \begin{eqnarray}j_t(X):=U^*_t\,X\,U_t=\phi(X,0)\label{6aa} \end{eqnarray}and the reflected flow\begin{eqnarray}r_t(X):=j_t(X\,J_t)=\phi(0,X)\label{6b} \end{eqnarray}where $X$ is a bounded system space operator.\begin{corollary} The Fermion flow $j_t(X)$ and the reflected flow  $ r_t(X)$ defined in  (\ref{6aa}) and (\ref{6b}) satisfy  the system of  quantum stochastic differential
 equations
 \begin{eqnarray}
 &dj_t(X)=j_t\left( i\,\left[H,X\right]-\frac{1}{2}\left(L^*LX+XL^*L-2L^*XL\right)\right)\,dt&\label{7}\\&+r_t\left(\left[L^*,X\right]\,W\right)\,dA_t +r_t\left(W^*\,\left[X,L\right]\right)\,dA_t^{\dagger}+j_t\left(W^*\,X\,W-X\right)\,d\Lambda_t& \nonumber\end{eqnarray}and\begin{eqnarray}&dr_t(X)= r_t\left( i\,\left[H,X\right]-\frac{1}{2}\left(L^*LX+XL^*L-2L^*XL\right)     \right)\,dt&\label{8}\\&-j_t\left(\{ L^*,X \}\,W \right)\,dA_t -j_t\left(W^*\{ X,L \}-2\,X\,L\right)\,dA_t^{\dagger}-r_t\left(W^*\,X\,W+X\right)\,d\Lambda_t& \nonumber\end{eqnarray}with\begin{eqnarray*}j_0(X)=r_0(X)=X.\end{eqnarray*}Here, as usual, $[x,y]:=x\,y-y\,x$ and $\{x,y\}:=x\,y+y\,x$.\end{corollary}\begin{proof} We replace $V_t$  and $V_t^*$ in Proposition \ref{prop1} by $U_t$ and $U_t^*$. Then, equation (\ref{7}) is a special case of (\ref{6}) for $T=X$, $S=0$ and\begin{eqnarray*}\alpha&=&-(iH+\frac{1}{2}\,L^*L)   \\\beta&=& -L^*\,W  \\\gamma&=&L
 \\\delta&=&W-1
 \end{eqnarray*}
 while equation (\ref{8}) is a special case of (\ref{6}) for $T=0$, $S=X$ and $\alpha,\beta,\gamma,\delta$ as above. \end{proof}\section{Generalized Fermion Flows} Fermion flows can be formulated and studied in a manner similar to Boson (also called Evans-Hudson) flows (cf. \cite{11a} and \cite{14}). Let $\mathcal{B}(\mathcal{H})$ denote the space of bounded system operators. We define  $\mathcal{B}(\mathcal{H})$-valued operations $\bigtriangledown $ and $\bigtriangleup $ on $\mathcal{B}(\mathcal{H})\times\mathcal{B}(\mathcal{H})$ by\begin{eqnarray}(T_1,S_1)\bigtriangledown (T_2,S_2):&=&T_1\,T_2+S_1\,S_2\label{op1}\\(T_1,S_1)\bigtriangleup (T_2,S_2):&=&T_1\,S_2+S_1\,T_2.\label{op2}\end{eqnarray}Notice that if\begin{eqnarray*}\rho(x,y):=(y,x)\end{eqnarray*}is the reflection map, then\begin{eqnarray*}(T_1,S_1)\bigtriangleup (T_2,S_2)=(T_1,S_1)\bigtriangledown \rho(T_2,S_2).\end{eqnarray*}We also define the
 $\mathcal{B}(\mathcal{H})\times\mathcal{B}(\mathcal{H})$-valued product map $\circ$ on $\left(\mathcal{B}(\mathcal{H})\times\mathcal{B}(\mathcal{H})\right)^2$ by\begin{eqnarray}x\circ y:= (x\bigtriangledown y, x \bigtriangleup y)=( T_1\,T_2+S_1\,S_2 , T_1\,S_2+S_1\,T_2  ) \label{op3}\end{eqnarray}where $x=(T_1,S_1)$ and $y=(T_2,S_2)$. From the definitions it follows that the $\circ$-product is associative with unit $\textrm{id}:=(1,0)$ where $1$ and $0$ are  the identity and zero  operators in $\mathcal{B}(\mathcal{H})$. Let the  flow $\phi_t(T,S):=V_t^*\,(T+S\,J_t)\,V_t$ be as in Proposition \ref{prop1} and let $x=(T_1,S_1)$ and
 $y=(T_2,S_2)$. Then
 \begin{eqnarray*}
 \phi_t(x)\,\phi_t(y)&=&\phi_t(T_1,S_1)\phi_t(T_2,S_2)\\&=&V_t^*\,(T_1+S_1\,J_t)\,V_t\,V_t^*\,(T_2+S_2\,J_t)\,V_t\\&=&V_t^*\,(T_1+S_1\,J_t)\,(T_2+S_2\,J_t)\,V_t\\&=&V_t^*\,(T_1\,T_2+S_1\,S_2+(T_1\,S_2+S_1\,T_2)\,J_t)\,V_t\\&=&\phi_t(T_1\,T_2+S_1\,S_2,T_1\,S_2+S_1\,T_2)\\&=&\phi_t( x\circ y )\end{eqnarray*}i.e $\phi_t$ is a homomorphism with respect to the $\circ$-product. Since $\phi_t$ is the solution of (\ref{6}), (\ref{6a}) this suggests considering flows satisfying quantum
 stochastic differential equations of the form
 \begin{eqnarray}
 d\phi_t(x)&=&\phi_t(\theta_1(x))\,dt+\phi_t(\theta_2(x))\,dA_t+\phi_t(\theta_3(x))\,dA_t^{\dagger}+\phi_t(\theta_4(x))\,
 d\Lambda_t\qquad\label{9}
 \end{eqnarray}
 with
 \begin{eqnarray*}
 \phi_0(x)= x  \bigtriangledown \textrm{id}  +x \bigtriangleup \textrm{id} \label{9a}
 \end{eqnarray*}
 where $x=(T,S)\in\mathcal{B}(\mathcal{H})\times\mathcal{B}(\mathcal{H})  $ and for $i=1,2,3,4$,
 $\theta_i:\mathcal{B}(\mathcal{H})\times\mathcal{B}(\mathcal{H})\longrightarrow\mathcal{B}(\mathcal{H})\times\mathcal{B}(\mathcal{H})$ are the "structure maps". In the case of (\ref{6}), (\ref{6a})\begin{eqnarray*}\theta_1(T,S)&=&(\alpha^*\,T+T\,\alpha+\gamma^*\,T\,\gamma,\alpha^*\,S+S\,\alpha+\gamma^*\,S\,\gamma)  \\\theta_2(T,S)&=&(\gamma^*\,S+S\,\beta-\gamma^*\,S\,(2+\delta),\gamma^*\,T+T\,\beta+\gamma^*\,T\,\delta)\\\theta_3(T,S)&=&(\beta^*\,S+S\,\gamma-\delta^*\,S\,\gamma,\beta^*\,T+T\,\gamma+\delta^*\,T\,\gamma)\\\theta_4(T,S)&=&(\delta^*\,T+T\,\delta+\delta^*\,T\,\delta,-(2\,S+S\,\delta+\delta^*\,S+\delta^*\,S\,\delta))\end{eqnarray*}where\begin{eqnarray*}\alpha&=&-(iH+\frac{1}{2}\,L^*L)   \\\beta&=& -L^*\,W  \\\gamma&=&L   \\\delta&=&W-1.\end{eqnarray*} The general conditions on the $\theta_i$ in order for $\phi_t$ to be an identity preserving $\circ$-product homomorphism are given in the following:\begin{proposition} Let $\phi_t$ be the solution of (\ref{9}). Then
 \begin{eqnarray*}
 \phi_t( x )\,\phi_t( y )&=&\phi_t( x\circ y )\\\phi_t( x )^*&=&\phi_t( x^* )\\\phi_t(\textrm{id})&=&1
 \end{eqnarray*}
 if and only if the $\theta_i$ satisfy the structure equations
 \begin{eqnarray}
 \theta_1(x)\circ y+x \circ \theta_1(y)+\theta_2(x)\circ\theta_3(y)&=&\theta_1(x\circ y)\label{s1}\\\theta_2(x)\circ y+x\circ\theta_2(y)+\theta_2(x)\circ\theta_4(y)&=&\theta_2(x\circ y)\label{s2}\\\theta_3(x)\circ y+x\circ\theta_3(y)+\theta_3(x)\circ\theta_4(y)&=&\theta_3(x\circ y)\label{s3}\\\theta_4(x)\circ y+x\circ\theta_4(y)+\theta_4(x)\circ\theta_4(y)&=&\theta_4(x\circ y)\label{s4}\end{eqnarray}and, with $*$ denoting "adjoint" and $x=(T_1,S_1)\Leftrightarrow
 x^*=(T_1^*,S_1^*)$,
 \begin{eqnarray}
 \theta_1(x^*)&=&(\theta_1(x))^*\label{s5}\\\theta_2(x^*)&=&(\theta_3(x))^*\label{s6}\\\theta_3(x^*)&=&(\theta_2(x))^*\label{s7}\\\theta_4(x^*)&=&(\theta_4(x))^*\label{s8}\end{eqnarray}with\begin{eqnarray}\theta_1(\textrm{id})=\theta_2(\textrm{id})=\theta_3(\textrm{id})=\theta_4(\textrm{id})=(0,0).\label{s9}\end{eqnarray}\end{proposition}\begin{proof} 
\begin{eqnarray*}
&\phi_t( x\circ y )=\phi_t( x )\,\phi_t( y )&\\
 &\Leftrightarrow d\phi_t( x\circ y )=d\phi_t( x )\,\phi_t( y )+\phi_t( x )\,d\phi_t( y )+d\phi_t( x )\,d\phi_t( y )&
\end{eqnarray*}

which implies that

\begin{eqnarray*}
&\phi_t(\theta_1(x\circ y  ))\,dt+\phi_t(\theta_2( x\circ y   ))\,dA_t+\phi_t(\theta_3( x\circ y   ))\,dA_t^{\dagger}&\\
&+\phi_t(\theta_4(     x\circ y
 ))\,d\Lambda_t=\phi_t(\theta_1(x))\,\phi_t(y)\,dt+\phi_t(\theta_2(x))\,\phi_t(y)\,dA_t&\\
&+\phi_t(\theta_3(x))\,\phi_t(y)\,dA_t^{\dagger}+\phi_t(\theta_4(x))\,\phi_t(y)\,d\Lambda_t+\phi_t(x)\phi_t(\theta_1(y))\,dt&\\
&+\phi_t(x)\,\phi_t(\theta_2(y))\,dA_t+\phi_t(x)\,\phi_t(\theta_3(y))\,dA_t^{\dagger}+\phi_t(x)\,\phi_t(\theta_4(y))\,d\Lambda_t+&\\&\phi_t(\theta_2(x))\,\phi_t(\theta_3(y))\,dt+\phi_t(\theta_2(x))\,\phi_t(\theta_4(y))\,dA_t&\\
&+\phi_t(\theta_3(x))\,\phi_t(\theta_4(y))\,dA_t^{\dagger}+\phi_t(\theta_4(x))\,\phi_t(\theta_4(y))\,d\Lambda_t\end{eqnarray*}from which collecting the $dt,dA_t,dA_t^{\dagger}$ and $d\Lambda_t  $ terms on each side, using the homomorphism property and then equating the coefficients of  $dt,dA_t,dA_t^{\dagger}$ and $d\Lambda_t  $   on both sides we obtain (\ref{s1})-(\ref{s4}). Similarly\begin{eqnarray*}\phi_t( x )^*=\phi_t( x^* )\Leftrightarrow d\phi_t( x )^*=d\phi_t( x^* )\end{eqnarray*}i.e
 \begin{eqnarray*}
 &\phi_t(\theta_1(x)^*)\,dt+\phi_t(\theta_2(x)^*)\,dA_t^{\dagger}+\phi_t(\theta_3(x)^*)\,dA_t+\phi_t(\theta_4(x)^*)\,d\Lambda_t=&\\&\phi_t(\theta_1(x^*))\,dt+\phi_t(\theta_2(x^*))\,dA_t+\phi_t(\theta_3(x^*))\,dA_t^{\dagger}+\phi_t(\theta_4(x^*))\,d\Lambda_t&\end{eqnarray*}and (\ref{s5})-(\ref{s8}) follow by equating the coefficients of  $dt,dA_t,dA_t^{\dagger}$ and $d\Lambda_t  $   on both sides. Finally \begin{eqnarray*}\phi_t(\textrm{id})=1\Leftrightarrow d\phi_t( \textrm{id} )=0\end{eqnarray*}which by the linear independence of $dt$ and
 the stochastic differentials implies (\ref{s9}).
 \end{proof}
 \section{Optimal Noise Coefficients}
 As in \cite{2} and \cite{4}, motivated by classical linear system control theory, we think of the self-adjoint
 operator $H$ appearing in (\ref{4}) as fixed and we consider the problem of  determining the coefficients $L$ and $W$
 of the noise part of the Hamiltonian of the evolution equation (\ref{4})  that minimize the "evolution performance
 functional"
 \begin{eqnarray}
 Q_{\xi,T}(u)=\int_0^T\,[\|X\,U_t\, \xi\|^2+\frac{1}{4}\,\|L^*\,L\,U_t\, \xi\|^2]\,dt+\frac{1}{2}\,\|L\,U_T\, \xi\|^2
 \label{q1}
 \end{eqnarray}
 where $T\in[0,+\infty)$, $\xi$ is an arbitrary vector in the exponential domain of $\mathcal{H}\otimes \Gamma$
 and $X$ is a bounded self-adjoint system operator. By the unitarity of $U_t$, $U_T$,$J_t$, and $J_T$, (\ref{q1})
 is the same as the "Fermion flow performance functional"
 \begin{equation}
 J_{\xi,T}(L,W)=\int_0^T\,[\,\|j_t(X)\,\xi\|^2+\frac{1}{4}\|j_t(L^*L)\,\xi\|^2\, ]\,dt+\frac{1}{2}\|j_T(L)\,\xi\|^2
 \label{q2}
 \end{equation}
 and the "reflected flow performance functional"
 \begin{equation}
 R_{\xi,T}(L,W)=\int_0^T\,[\,\|r_t(X)\,\xi\|^2+\frac{1}{4}\|r_t(L^*L)\,\xi\|^2\, ]\,dt+\frac{1}{2}\|r_T(L)\,\xi\|^2
 \label{q3}
 \end{equation}
 We consider the problem of  minimizing the  functionals $ J_{\xi,T}(L,W)$ and   $R_{\xi,T}(L,W)$   over all system
 operators $L,W$ where $L$ is bounded and $W$ is unitary. The motivation behind
 the definition of the performance functionals (\ref{q1}), (\ref{q2})  and (\ref{q3}) can be found in the following
 theorem which is the quantum stochastic analogue of the classical linear regulator theorem.
 \begin{theorem}
 Let $U=\{U_t\,,\,t\geq 0\}$ be a  process satisfying the  quantum stochastic differential equation
 \begin{equation}
 dU_t=(F\,U_t+u_t)\,dt+ \Psi \,U_t\, dF_t+ \Phi \,U_t\, dF_t^{\dagger}+Z\,U_t\,d\Lambda_t,\,U_0=1,\,t\in [0,T]
 \label{f1}
 \end{equation}
 with adjoint
 \begin{equation}
 dU_t^*=(U_t^*\,F^*+u_t^*)\,dt+ U_t^*\,\Psi^* \,dF_t^{\dagger}+ U_t^*\,\Phi^* dF_t+U_t^*\,Z^*\,d\Lambda_t,\,U_0^*=1
 \label{f2}
 \end{equation}
 where $0<T<+\infty$,  the coefficients $F,\,\Psi,\,\Phi, \,Z$  are bounded operators on the system space
 $\mathcal{H}$ and $u_t:=-\Pi \,U_t$ for some positive bounded system operator $\Pi$.  Then the functional
 \begin{eqnarray}
 Q_{\xi,T}(u)=\int_0^T\,[<U_t \,\xi,X^2\,U_t\, \xi>+<u_t\, \xi,u_t\, \xi>]\,dt-<u_T\, \xi,U_T \,\xi>
 \label{q4}
 \end{eqnarray}
 where $X$  is a system space observable, identified with its ampliation $X \otimes I$ to $\mathcal{H}  \otimes \Gamma$,
 is minimized over the set of feedback control processes of the form $u_t=-\Pi \,U_t$, by choosing $\Pi$ to be a
 bounded,  positive, self-adjoint   system operator satisfying
 \begin{eqnarray}
 \Pi \, F+F^* \Pi+{\Phi}^* \Pi \Phi-{\Pi}^2+X^2&=&0\label{r}\\\Pi \,\Psi + {\Phi}^* \, \Pi + {\Phi}^* \, \Pi \, Z&=&0
 \label{r1}\\\Pi \,Z +Z^* \,\Pi + Z^* \, \Pi \,Z&=&0\label{r2}.
 \end{eqnarray}
 The minimum value is $<\xi,\Pi \xi>$. We recognize (\ref{r}) as the algebraic Riccati equation.
 \end{theorem}
 \begin{proof} Let
 \begin{eqnarray}
 \theta_t=<\xi,U_t^*\,\Pi \,U_t\,\xi>\label{r3}.
 \end{eqnarray}
 Then
 \begin{eqnarray}
 d\theta_t=<\xi,d(U_t^* \,\Pi \,U_t)\,\xi>=<\xi,(dU^*_t \, \Pi \,U_t+U_t^* \,\Pi \,dU_t+dU_t^* \,\Pi \,dU_t)\,\xi>
 \label{r4}
 \end{eqnarray}
 which, after replacing $dU_t$ and $dU_t^*$ by (\ref{f1}) and (\ref{f2}) respectively,  and using (\ref{1}), (\ref{2})
 and the  It\^{o} table, becomes
 \begin{eqnarray}
 &d\theta_t=<\xi, U_t^*\,((F^*\,\Pi+\Pi\,F+\Phi^*\,\Pi\,\Phi)\,dt+(\Phi^*\,\Pi+\Pi\,\Psi+\Phi^*\,\Pi\,Z)\,J_t\,dA_t\qquad&
 \label{r5}
 \\&+(\Psi\,\Pi^*+\Pi\,\Phi+Z^*\,\Pi\,\Phi)\,J_t\,dA_t^{\dagger}+(Z^*\,\Pi+\Pi\,Z+Z^*\,\Pi\,Z)\,d\Lambda_t)\,U_t\,\xi>&
 \nonumber\\
 &+<\xi,(u_t^*\,\Pi\,U_t+U_t^*\,\Pi\,u_t)\,dt\,\xi>.
 &\nonumber
 \end{eqnarray}
 and by (\ref{r})-(\ref{r2})
 \begin{eqnarray}
 &d\theta_t=<\xi, U_t^*\,( \Pi^2-X^2)\,U_t\,dt\,\xi>   +<\xi,(u_t^*\,\Pi\,U_t+U_t^*\,\Pi\,u_t)\,dt\,\xi>&
 \label{r6}
 \end{eqnarray}
 By (\ref{r3})
 \begin{eqnarray}
 &\theta_T-\theta_0=<\xi, U_T^*\,\Pi\,U_T\,\xi>-<\xi,\Pi\,\xi>.&\label{r7}
 \end{eqnarray}
 while by (\ref{r6})
 \begin{eqnarray}
 &\theta_T-\theta_0=&\label{r8}\\
&\int_0^T\,(<\xi, U_t^*\,( \Pi^2-X^2)\,U_t\,\xi>   +<\xi,(u_t^*\,\Pi\,U_t+U_t^*\,\Pi\,u_t)\,\xi>)\,
 dt\qquad&\nonumber
 \end{eqnarray}
 By (\ref{r7}) and (\ref{r8})
 \begin{eqnarray}&<\xi, U_T^*\,\Pi\,U_T\,\xi>-<\xi,\Pi\,\xi>=&
 \label{r9}\\&\int_0^T\,(<\xi, U_t^*\,( \Pi^2-X^2)\,U_t\,\xi>
 +<\xi,(u_t^*\,\Pi\,U_t+U_t^*\,\Pi\,u_t)\,\xi>)\,dt.&\nonumber
 \end{eqnarray}
 Thus
 \begin{eqnarray}
 &Q_{\xi,T}(u)=(<\xi, U_T^*\,\Pi\,U_T\,\xi>-<\xi,\Pi\,\xi>)+Q_{\xi,T}(u)&\label{r10}\\&- (<\xi, U_T^*\,\Pi\,U_T\,\xi>-<\xi,\Pi\,\xi>).&\nonumber\end{eqnarray}Replacing the first parenthesis on the right hand side of (\ref{r10}) by (\ref{r9}), and $Q_{\xi,T}(u)$by (\ref{q4}) we obtain after
 cancelations
 \begin{eqnarray}
 Q_{\xi,T}(u)&=&\int_0^T \,(<\xi,(U_t^* \,{\Pi}^2\,U_t+u_t^* \,\Pi \,U_t+U_t^* \,\Pi \,u_t+u_t^* \,u_t)\,\xi>\,
 dt\nonumber\\
&+&<\xi,\Pi\,\xi>\label{r11}
 \\
 &=&\int_0^T \,||(u_t+\Pi\,U_t)\,\xi||^2\,dt+<\xi,\Pi\,\xi>\nonumber
 \end{eqnarray}
 which is clearly minimized by $u_t=-\Pi\,U_t$ and the minimum is $<\xi,\Pi\,\xi>$.
 \end{proof}
 \begin{theorem}
 Let $X$ be a bounded self-adjoint system operator such that the pair ($i\,H$, $X$) is stabilizable
 i.e there exists a bounded system operator $K$ such that $i\,H+KX$ is the generator of an asymptotically stable
 semigroup.  Then, the quadratic  performance functionals
 (\ref{q2})
 and (\ref{q3}) associated with the Fermion flow $\{j_t(X):=U_t^*\,X \,U_t\,,\,t \geq 0\}$  and the reflected
 flow $\{r_t(X):=j_t(X\,J_t)\,,\,t \geq 0\}$,  where $U=\{U_t\,,\,t\geq 0\}$ is the solution of (\ref{4}), are
 minimized by
 \begin{eqnarray}
 &L=\sqrt{2}\,\Pi^{1/2}\,W_1\label{r12}&
 \end{eqnarray}
 and
 \begin{eqnarray}&W=W_2&\label{r13}
 \end{eqnarray}where  $\Pi$ is a positive self-adjoint solution of the ``algebraic Riccati equation''
 \begin{eqnarray}
 &i\,[H,\Pi]+\Pi^2+X^2=0& \label{r14}
 \end{eqnarray}
 and $W_1$, $W_2$ are bounded unitary system operators  commuting with  $\Pi$.  It is known (see \cite{15}) that if the pair ($i\,H$, $X$) is stabilizable, then (\ref{r14}) has a positive self-adjoint solution $\Pi$.  Moreover \begin{eqnarray}&\min_{L,W}\, J_{\xi,T}(L,W)= \min_{L,W}\, R_{\xi,T}(L,W)=<\xi,\Pi\,\xi>&\label{r15}\end{eqnarray} independent of $T$.\end{theorem}\begin{proof}Looking at (\ref{4}) as (\ref{f1}) with
 $u_t=-\frac{1}{2}\,L^*\,L\,U_t$, $F=-i\,H$, $\Psi=-L^*\,W$, $\Phi=L$, and $Z=W-1$,
 (\ref{q4}) is identical to (\ref{q1}). Moreover, equations  (\ref{r})-(\ref{r2}) become
 \begin{eqnarray}
 &i\,[H,\Pi]+L^*\,\Pi\,L-\Pi^2+X^2=0&\label{r16}\\&L^*\,\Pi-\Pi\,L^*\,W+L^*\,\Pi\,(W-1)=0&\label{r17}\\&(W^*-1)\,\Pi+\Pi\,(W-1)+(W^*-1)\,\Pi\,(W-1)=0.\label{r18}&\end{eqnarray}By the self-adjointness of $\Pi$, (\ref{r17}) implies that
 \begin{eqnarray}
 &[L,\Pi]=[L^*,\Pi]=0&\label{r19}
 \end{eqnarray}
 while (\ref{r18}) implies that
 \begin{eqnarray}
 &[W,\Pi]=[W^*,\Pi]=0&.\label{r20}
 \end{eqnarray}
 i.e (\ref{r13}).
 By (\ref{r17}) and the fact that in this case
 \begin{eqnarray}
 &\Pi=\frac{1}{2}\,L^*\,L \mbox{  i.e  }L^*\,L=2\,\Pi&\label{r21}
 \end{eqnarray}

 equation (\ref{r16}) implies (\ref{r14}). Equations (\ref{r19}) and (\ref{r21}) also imply that
 \begin{equation}
 [L,L^*]=0\label{r22}
 \end{equation}
 which implies (\ref{r12}).
 \end{proof}
 
 \end{document}